\documentclass[letterpaper, 10 pt, conference]{ieeeconf}  

\IEEEoverridecommandlockouts                              
\usepackage{answers}
\usepackage{setspace}
\usepackage{graphicx}
\usepackage{multicol}
\usepackage{mathrsfs}
\usepackage{amsmath,amssymb}
\usepackage{bbold}
\usepackage{tikz}
\usetikzlibrary{backgrounds}
\usepackage[framemethod=tikz]{mdframed}
\usepackage{fancybox}

\usepackage[mathscr]{euscript}

\usepackage{schemabloc}

\usetikzlibrary{circuits}

\usepackage[english]{babel}
\usepackage[T1]{fontenc}

\usetikzlibrary{arrows, positioning, calc}

\usepackage{blox}

\usepackage{graphicx} 
\usepackage{epsfig} 
\usepackage{pgfplots}
\usepackage{pgfplotstable}
\usepackage{resizegather}
\usepackage{mdframed}

\mdfsetup{linewidth=0.75pt,
innertopmargin=.25pt,
  skipbelow=.25em,
  skipabove=.25em,
  nobreak=true
}

\usepackage{amsmath} 
\usepackage{amssymb}  
\usepackage{xcolor}
\usepackage{mathtools}
\usepackage{lipsum}
\usepackage{soul}
\usepackage{relsize}

\usepackage{amsthm}
\usepackage[ruled]{algorithm2e}
\LinesNumbered

\newenvironment{subproof}[1][]{\begin{proof}[#1]}{\end{proof}}
\makeatletter
\newcommand{\removelatexerror}{\let\@latex@error\@gobble}
\makeatother
\allowdisplaybreaks[1]
\usepackage{balance}

\usepackage[normalem]{ulem}
\usepackage{xfrac}

\newtheorem{theorem}    {Theorem}
\newtheorem{prop} {Proposition}

\newtheorem{claim}    {Claim}

\newtheorem{remark}     {Remark}

\theoremstyle{definition}

\newtheorem{assm} {Assumption}

\newcommand{\nc}{\newcommand}
\nc{\R}{{\mathbb R}}
\nc{\C}{{\mathbb C}}
\nc{\Z}{{\mathbb Z}}
\nc{\N}{{\mathbb N}}
\nc{\s}{\bar{\bf s}}
\nc{\I}{{\cal I}^*}
\nc{\e}{\boldsymbol{\mathsf e}}
\nc{\ellb}{\boldsymbol\ell}
\nc{\be}{\begin{equation}}
\nc{\ee}{\end{equation}}
\nc{\uk}{\underline{k}}
\nc{\bk}{\bar{k}}

\usepackage{pst-sigsys}

\let\labelindent\relax
\usepackage{enumitem}   

\usepackage{multicol,lipsum}

\usepackage{float}

\usepackage{nicematrix}

\usepackage{textcomp}

\usepackage{cite}

\usepackage{upgreek}

\usepackage{multibib}

\usepackage[pagebackref=false]{hyperref} 
\hypersetup{
    plainpages=false,       
    unicode=false,          
    pdftoolbar=true,        
    pdfmenubar=true,        
    pdffitwindow=false,     
    pdfstartview={FitH},    
    pdftitle={Adaptive Control with Set-Point Tracking and Linear-like Closed-loop Behavior},    
   pdfauthor={Mohamad T. Shahab},    
    pdfnewwindow=true,      
    colorlinks=true,        
    linkcolor=black,         
    citecolor=black,        
    filecolor=magenta,      
    urlcolor=cyan          
}

\makeatletter 
\pretocmd\@bibitem{\color{black}\csname keycolor#1\endcsname}{}{\fail}
\newcommand\citecolor[1]{\@namedef{keycolor#1}{\color{black}}}
\makeatother

\title{\LARGE \bf 
Adaptive Control with Set-Point Tracking and Linear-like Closed-loop Behavior
}

\author{Mohamad T. Shahab
\thanks{The author
        is with 
        the Department of Electrical, Computer, and Biomedical Engineering,
        Toronto Metropolitan University, 
        Toronto, ON M5B 2K3, Canada.
        Email: {\tt mshahab@torontomu.ca}.
        }
\thanks{
  This work was supported by the faculty
  start-up fund provided by the 
  Toronto Metropolitan University.
}
      }
\begin{document}

\baselineskip=12.25pt
\maketitle
\thispagestyle{empty}
\pagestyle{plain}

\begin{abstract}
In this paper, we consider the problem of set-point tracking
for a discrete-time plant with unknown
plant parameters belonging to a convex and compact uncertainty
set. 
We carry out parameter estimation
for an associated auxiliary plant, and a pole-placement-based
control law is employed.
We prove that this adaptive
controller provides 
desirable linear-like closed-loop behavior
which guarantees a bound 
consisting of:
exponential decay with respect to the initial condition, 
a linear-like convolution bound with respect to the exogenous inputs, and
a constant 
scaled by the square root of the constant in the denominator 
of the parameter estimator update law.
This implies that the system has a bounded gain.
Moreover, asymptotic tracking is 
also proven when the disturbance is constant.

\end{abstract}

\section{Introduction}

Adaptive control is an approach used to deal with systems 
with uncertain and/or time-varying 
parameters. Initial general proofs came around 1980, e.g. see \cite{morse1978}, \cite{Goodwin1980}, \cite{Morse1980}, \cite{Narendra1980} and \cite{Narendra1980_pt2}.
These original adaptive controllers
are typically not robust to unmodelled dynamics, do not tolerate time-variations very well,
may have poor transient behavior and do not handle noise/disturbances well (e.g.
see \cite{rohrs}). 
Afterwards, some controller design changes were proposed
to deal with the above shortcomings, such as the use of signal normalization, 
deadzones and $\sigma$-modification, e.g.
see \cite{Ioa86}, \cite{kreiss}, \cite{rick2}, \cite{rick}, 
and \cite{Tsakalis4};
also, simply using projection onto a convex set
of admissible parameters turned out to be powerful, 
e.g. see \cite{hanfu}, \cite{Naik}, \cite{Wen}, \cite{Wenhill} and \cite{ydstie}. 
However, in general, 
these redesigned
controllers 
may provide asymptotic 
closed-loop
behavior
but no exponential stability nor bounded
gain on the noise are proven\footnote{In 
Ydstie's \cite{ydstie}, an exception, a bounded gain is proven.}.
Although, some controllers, especially those using projection,
provide a bounded-noise bounded-state property, 
as well as tolerance of some degree
of unmodelled dynamics and/or time-variations.

Recently, 
for
discrete-time
LTI plants,
in the $d$-step ahead control and the model reference adaptive control (MRAC) setting \cite{scl17}, \cite{acc19}, \cite{mcss20}, \cite{cdc21}, 
and in the pole-placement control setting \cite{ccta17}, \cite{mcss18}, \cite{cdc19}, \cite{TAC22},  
subject to some standard assumptions,
a new approach has been proposed
which not only provides exponential
stability and a bounded gain on the noise, but also a
linear-like
convolution bound on the exogenous inputs;
this convolution bound is leveraged to prove 
tolerance 
to a degree of time-variations and to a degree of unmodelled
dynamics \cite{acs24}. 
The key idea is to use the
original (ideal) projection algorithm in conjunction with a restriction
of the parameter estimates to a convex set,
although
this convexity requirement was relaxed in
\cite{cdc18}, \cite{mcss18}, \cite{cdc19} and \cite{TAC22}.
However, in the literature it is very common  
to adopt a modified
version of the ideal projection algorithm in which a constant is added
to the denominator of the estimator update law. 
This is widely used, e.g. in \cite{kreiss}, \cite{kreiss2}, \cite{hanfu}, 
\cite{rick}, \cite{GOODWIN1987}, \cite{rick2},
\cite{ydstie}, and \cite{hyst92}, and in
the disturbance-free case, stability
and tracking is proven under minimal assumptions
 \cite{Goodwin1980}.
In the presence of disturbances, however, analysis is harder;
for example, in \cite{Wen} 
it is proven that a bounded disturbance yields a bounded state.
A linear-like result was shown in our paper \cite{tac25} utilizing the modified algorithm, but it only considered the MRAC setting for minimum phase plants.

In this paper, our objective is to provide set-point (step) tracking and
to obtain
a quantitative bound in terms of the initial condition, the exogenous inputs,
and the key estimator constant parameter. 
This is done
by
extending the ideas of the recent work on
the ideal projection algorithm to the aforementioned more widely-used modified parameter estimator.
To this end, we impose standard classical assumptions, and
since our objective is
to obtain uniform bounds, 
we impose a natural compactness assumption.
Here we prove a bound on the closed-loop behavior consisting of
three terms:
a decaying exponential on the initial condition,
a linear-like convolution bound on the noise input, and
a constant 
scaled by the square root of the constant in the denominator
of the estimator update law and the set-point.
This closed-loop behavior is clearly stronger than showing a ``bounded-noise
bounded-state'' performance. 
We would like also to mention that there is no persistent
excitation requirement of any sort in the closed-loop system.
Furthermore, for a constant disturbance entering the system, 
we prove asymptotic tracking (and disturbance rejection) as well.
The controllers in our previous work \cite{cdc19}, \cite{TAC22} 
provide set-point tracking and linear-like bounds, however 
the modified estimator was not used and they
only consider a deadbeat control design which we relax and generalize here.

{\bf Notation.} 
We denote $\R$, $\Z$, $\Z^+ $, $\N $ and $\C$ as the set of real numbers, integers, non-negative integers, natural numbers and complex numbers, respectively.
 We will denote the Euclidean-norm of a vector and the induced norm of a matrix by 
the subscript-less default notation $\|\cdot\|$. 
Also, $\ellb_{\infty} $ denotes the set of bounded sequences.
For a signal $f\in\ellb_\infty $, define the $\infty $-norm by
$\|f \|_\infty:= \sup_{t\in \Z} | f(t) | $.
For a closed and convex set $\Omega\subset\R^p$, let the function
 $\mathrm{Proj}_{\Omega}
 \left\{ \cdot \right\}:\R^{p}
 \rightarrow \Omega $ 
denote the projection 
onto the set $\Omega  $
in the 2-norm;
because the set $\Omega $ is closed and convex, 
the function $\mathrm{Proj}_{\Omega} $ is well-defined.
Let $\mathbf 0_{p\times q} $ denote the $p \times q$ matrix
whose entries are all zeros, and 
$I_{p} $ denote the
identity matrix of size $p$.
Define the normal vector $\e_j \in\R^{p}$ 
of appropriate length $p$ by
  $$\e_j:=
      \bigl[
    \underbrace{\begin{matrix}0 & \cdots & 0 \end{matrix}}_{j-1\text{ elements} } \;\; \begin{matrix}1 & 0 & \cdots & 0 \end{matrix}
      \bigr]^\top.$$
For a signal
$f$ which is 
sufficiently well-behaved to have a $z$-transform,
we let $F(z) $ denote this quantity.

\section{The Setup}

We consider the ${n{{\text{th}}}}$-order linear time-invariant discrete-time plant
\begin{flalign}
y(t+1)
&=
\sum_{j=1}^n
a_{j} y(t-j+1) 
+ 
\sum_{j=1}^n b_{j}u(t-j+1)
+w(t), 
\nonumber \\
&\qquad t\in\Z,
\label{plant1}
\end{flalign}
with $y(t),u(t),w(t)\in\R $ denoting the measured output,
 the control input, and the disturbance/noise input, respectively.
We can represent the plant model by the vector of parameters
$$
\theta=
\begin{bmatrix}
a_1 & a_2 & \cdots & a_{n} & b_1 & b_2 & \cdots & b_{n}
\end{bmatrix}^\top
\in 
{\cal S} \subset
\R^{2n}.$$
We assume that $\theta$ is unknown but the set ${\cal S}\subset \R^{2n}$ is known.
For every $\theta\in{\cal S}$, associated with the plant \eqref{plant1}
are the polynomials
\[
  \mathbf A_\theta(z^{-1}):=1-a_1 z^{-1}-a_2 z^{-2}\cdots-a_n z^{-n},
\]
and
\[
 \mathbf B_\theta(z^{-1}):=b_1 z^{-1}+b_2 z^{-2}\cdots+b_n z^{-n},
\]
and the strictly proper 
transfer function 
$\frac{\mathbf B_\theta(z^{-1})}{\mathbf A_\theta(z^{-1})}$.
Such a plant can be expressed in the (two-sided) $z$-transform form as
\begin{equation}
\label{plantTF1}
\mathbf A_\theta(z^{-1})Y(z)=\mathbf B_\theta(z^{-1})U(z)+z^{-1}W(z).
\end{equation}
We impose an assumption on the set of admissible plant parameters.
\begin{assm}
\label{assume1}
${\cal S} $ is convex and compact, and for each $\theta \in {\cal S}$,
the polynomials 
$z^n\mathbf A_{\theta}(z^{-1} ) $ and $z^n\mathbf B_{\theta}(z^{-1} )$ are coprime.
\end{assm}
\noindent
The convexity part of the above assumption is
common in the adaptive control literature---it is used to
facilitate parameter projection, e.g.  see \cite{goodwinsin}.
The boundedness part 
is quite
reasonable in practical situations; it is used here to ensure that
we can prove uniform bounds and decay rates on the closed-loop behavior.

Our objectives here is to prove an exponential form of
stability, 
a convolution bound as well as a bounded gain on the noise,
and
asymptotic tracking of a desired set-point 
(and constant disturbance rejection).
To achieve the tracking objective, we require the following assumption:
\begin{assm} \label{assm2}
For every $\theta\in{\cal S}$, the corresponding polynomial $\mathbf B_\theta(z^{-1} )$ is such that $\mathbf B_\theta(1 )\neq0$.
\end{assm}
\begin{remark}
Here we allow the plant to have zeros outside
the open unit desk. Hence, the plant could be unstable and non-minimum phase, which makes it challenging to control.
\end{remark}

In this paper we will do system identification on a related auxiliary model rather than the original plant model. 
Let $y^*(t)= r \in\R,\,t\in\Z $, be the desired constant 
reference signal; so in the $z$-domain, this class of reference signals is described by $(1-z^{-1})Y^*(z)=0$.
Let us define the tracking error by
\begin{equation}
\label{tracking}
\overline y(t):=y(t)-y^*(t),
\end{equation}
and an auxiliary control input by
\begin{equation}
\label{ubar_def}
\overline u(t):=u(t)-u(t-1)
\end{equation}
as well as an adjusted disturbance signal by
\begin{equation}
\label{wbar_def}
\overline w(t):=w(t)-w(t-1).
\end{equation}
Accordingly, if we then multiply both sides of the plant model \eqref{plantTF1} by $(1-z^{-1})$ and then subtract $(1-z^{-1})\mathbf A_\theta(z^{-1})Y^*(z) $ from both sides, 
we obtain 
\begin{equation}
\underbrace{(1-z^{-1})\mathbf A_\theta(z^{-1})}_{=:\overline{\mathbf A}_\theta(z^{-1})}\overline Y(z) 
= 
\mathbf B_\theta(z^{-1}) 
\overline U(z)
+
z^{-1} 
  \overline W(z)
  .
  \label{plantTF2}
\end{equation}
Notice that the polynomial $\overline {\mathbf A}_\theta(z^{-1})$ has the form
\begin{flalign*}
\overline {\mathbf A}_\theta(z^{-1})
&=
(1-z^{-1})\mathbf A_\theta(z^{-1}) 
 \\
&=
1- \underbrace{(1+a_1)}_{=:\overline a_1} z^{-1} - 
\underbrace{(a_2-a_1)}_{=:\overline a_2} z^{-2} - \cdots 
\\
&\qquad \cdots
- \underbrace{(a_n-a_{n-1})}_{=:\overline a_n} z^{-n} 
- \underbrace{(-a_n)}_{=:\overline a_{n+1}} z^{-(n+1)}.
\end{flalign*}
We see that the parameters of $\overline{\mathbf A}_\theta(z^{-1})$ are determined in a simple way from those of $\mathbf A_\theta(z^{-1})$. 
Indeed, we can form 
a invertible matrix ${\cal V}_n\in\R^{(2n+1)\times(2n+1)}$, 
\[
{\cal V}_n 
=
{
 \left[\begin{matrix}
1 & 1 &  &  &  & 
\\ 
 & -1  & 1 &  & &
\\
 &  & -1 & 1 &  
\\
 &  & & \ddots & \ddots & 
\\ 
 &  &  & & \ddots & 1
\\ 
 &  & &  &   & -1
\end{matrix}\right],
}
\]
depending solely on $n$,
so that
the parameters of $\overline{\mathbf A}_\theta(z^{-1}) $ and $\mathbf B_\theta(z^{-1}) $ 
of \eqref{plantTF2} are given by
\begin{equation*}
{\cal V}_n
{
\begin{bmatrix}
1 \\ a_1 \\ \vdots \\ a_{n} \\
b_1 \\ \vdots \\ b_{n}
\end{bmatrix}}
=
{
\begin{bmatrix}
\overline a_1 \\ \overline a_2 \\  \vdots \\ \overline a_{n+1} \\
b_1 \\ \vdots \\ b_{n}
\end{bmatrix}
}
=:\theta^*
.
\end{equation*}
So the set 
of admissible parameters of \eqref{plantTF2} is given by
\begin{equation}
\overline {\cal S}
 :=
{
\left\{ 
{\cal V}_n
\begin{bmatrix}
1 \\ \theta
\end{bmatrix}
: \; \theta
 \in {\cal S}
 \right\}
 \subset \R^{2n+1}
 }.
 \label{const_multi}
\end{equation}
Using this notation, the auxiliary plant \eqref{plantTF2} 
can now be put into the regressor form:
\begin{equation}
\label{plant2}
\overline y (t+1)=\psi(t)^\top \theta^*+\overline w(t),
\end{equation}
with 
$\theta^* \in \overline{\cal S} $
and
$\psi(t)\in\R^{2n+1}$ defined as
\begin{flalign*}
\psi(t)
&:=
\left[ 
\begin{matrix}
\overline y(t) & \overline y(t-1) & \cdots & \overline y(t-n)
\end{matrix}\right.
\\
&\qquad\qquad\quad
\left.
\begin{matrix} 
& \overline u(t) & \overline u(t-1)& \cdots & \overline u(t-n+1)
\end{matrix}
\right]
^\top.
\end{flalign*}

Since ${\cal S}$ is compact, it follows that $\overline{\cal S}$ is so as well. 
Furthermore, since ${\cal S}$ is convex and $\overline{\cal S}$ is a linear image of ${\cal S}$, then $\overline{\cal S}$ is convex as well \cite{convex_book}.
Additionally, because of Assumptions \ref{assume1} and \ref{assm2} we see that for every $\theta^*\in\overline {\cal S}$, the corresponding polynomials $z^{n+1}\overline {\mathbf A}_{\theta}(z^{-1} ) $ and $z^n\mathbf B_{\theta}(z^{-1} )$ are coprime and $\mathbf B_\theta(1)\neq0$.

In most adaptive controllers the goal is to prove asymptotic behavior, so details of initial conditions are not important. However, we want to get a bound on the transient behavior, in particular a bound including exponential decay with respect 
to the initial condition and a convolution sum bound with respect to the exogenous inputs. To this end, with a starting time of $t_0$, define the initial condition
\[
  \boldsymbol\phi_0:= 
  \begin{bmatrix}
  y(t_0) & \cdots & y(t_0-n) & u(t_0) & \cdots & u(t_0-n)
  \end{bmatrix}^\top.
\]
Clearly this provides $\psi(t_0)$ as well.

\section{The Adaptive Controller}

\subsection{Parameter estimation}

For the plant model \eqref{plant2},
starting with an initial estimate $\theta_0\in\overline{\cal S} $ at time $t_0 $,
given an estimate $\hat\theta(t) $ of $\theta^* $ at time $t\geq t_0 $,
a common way to obtain a new estimate is to solve the optimization problem
\begin{equation}
\hat\theta(t+1)=
\underset{X}{\mathrm{argmin}} \left\{
\|X-\hat\theta(t) \| : 
\overline y(t+1)=\psi(t)^\top X
\right\},
\nonumber
\end{equation}
which yields the {\bf ideal projection algorithm}:
\begin{equation}
\hat\theta(t+1)
 =
\begin{cases}
 \hat\theta(t)
 & \psi(t)=0 
 \\
 \hat\theta(t)
 +
  \displaystyle
 \frac{\psi(t)}{\|\psi(t)\|^2}
 \bigl[
  \overline y(t+1) - \psi(t)^\top \hat\theta(t)
  \bigr] 
  & \text{otherwise.}
  \end{cases}
  \label{orig1}
\end{equation}
We can also constrain the new estimates to $\overline{\cal S}$ by projection. Of
course, if $\|\psi(t)\|$ is close to zero, numerical problems
may occur, so it is the norm in the literature (e.g. \cite{goodwinsin} and \cite{Goodwin1980})
to add a constant to the denominator:\footnote{
In \cite{akhtar},
  the ideal algorithm \eqref{orig1} is used and
  Lyapunov stability is proven, 
  but a convolution bound on the exogenous inputs is not proven, and the high-frequency gain is assumed to be known exactly.
}
with $\mu>0$, consider the {\bf classical estimator}
(with associated projection onto $\overline{\cal S} $):
\begin{subequations}
\label{est2}
\begin{flalign}
\label{est2a}
 &\check\theta(t+1)
 =
 \hat\theta(t)
 +
  {\frac{\psi(t)}{\mu+\|\psi(t)\|^2}}
  \bigl[
  \overline y(t+1) - \psi(t)^\top \hat\theta(t)
  \bigr] ,
 \\
\label{est2b}
 &\hat\theta(t+1)
 =
\mathrm{Proj}_{\overline{\cal S}} \left\{ \check\theta(t+1) \right\}.
\end{flalign}
\end{subequations}
The ideal estimator \eqref{orig1} has been analyzed in great detail in our earlier work on adaptive control 
including 
the first-order one-step-ahead setup \cite{scl17},
the high-order $d$-step-ahead setup \cite{acc19}, \cite{mcss20},
the model reference setup \cite{cdc21},
the pole-placement stability problem \cite{mcss18}, 
and various extensions including multi-estimators
and switching \cite{cdc18}, \cite{cdc19}, \cite{TAC22}.
In all of these cases, quite surprisingly we are
able to prove, under suitable assumptions, 
that the closed-loop system exhibits linear-like
behavior.
Here we will focus on the more widely used classical estimator \eqref{est2} and prove that the corresponding closed-loop system exhibits 
linear-like behavior with an offset.

Define the prediction error by
\begin{equation} \label{predict1}
e(t+1):=\overline y(t+1) - \psi(t)^\top \hat\theta(t),
\end{equation}
and 
the parameter error 
by $\tilde\theta(t) := \hat\theta(t)-\theta^* $.
The following Proposition lists properties 
of the parameter estimator.
\begin{prop}[\hspace{-.2pt}\textbf{\cite{tac25}}]
\label{est_prop}
For every 
 $t_0\in\Z$, initial condition $ \boldsymbol \phi_0\in\R^{2(n+1)} $, $
\theta_0, \theta^*\in\overline{\cal S},
\overline w\in\ellb_\infty
 $,
and $\mu>0 $,
when the parameter estimator \eqref{est2} is applied to the plant \eqref{plant2}:
\\ 
(i) the following holds:
\begin{flalign}
&\|\tilde\theta(t) \|^2
\leq
\|\tilde\theta(\tau) \|^2
+
\sum_{j=\tau}^{t-1}
\biggl[
- \frac{1}{2}\frac{e(j+1) ^2 }{\mu+ \|\psi(j) \|^2 }
+
\nonumber
\frac{
2
\overline w(j)^2 }{\mu+ \|\psi(j) \|^2 }
\biggr]
,
\\
&\qquad\qquad\qquad
\qquad
t>\tau\geq t_0;
\nonumber
\end{flalign}
(ii)
on every interval of the form $[\underline t, \overline t ) \subset [t_0,\infty) $ which satisfies
$\|\psi(t) \|^2 \geq \mu, 
\,
t\in [\underline t, \overline t ),$
it follows that
\begin{flalign}
&\|\hat\theta(t+1)-\hat\theta(t) \|
\leq
\frac{|e(t+1) |}{\|\psi(t) \|},
\quad 
t
\in [\underline t, \overline t ),
\nonumber
\\
&\|\tilde\theta(t) \|^2
\leq
\|\tilde\theta(\tau) \|^2
+
\sum_{j=\tau}^{t-1}
\biggl[
- \frac{1}{4}\frac{e(j+1) ^2 }{ \|\psi(j) \|^2 }
+
\nonumber
\frac{2\overline w(j)^2 }{ \|\psi(j) \|^2 }
\biggr],
\\
&\qquad\qquad\qquad
\qquad
\overline t \geq t>\tau\geq \underline t.
\nonumber
\end{flalign}
\end{prop}

\subsection{Control law}
\unskip

The elements of $\hat{\theta}(t)$ can be partitioned naturally as
$\hat{\theta}(t)=:[\hat{\overline a}_{1}(t) 
\;\; \cdots \;\; \hat{\overline a}_{n+1}(t) 
\;\; \hat b_{1}(t) \;\; \hat b_{2}(t)\;\; \cdots 
\;\; \hat b_{n}(t)]^\top;$
we can define polynomials associated with these estimates:
$$\widehat {\overline {\mathbf A}}_{\hat{\theta}(t)}(z^{-1}):=1-\hat{\overline a}_{1}(t)z^{-1}-
\hat {\overline a}_{2}(t)z^{-2}\cdots-\hat {\overline a}_{n+1}(t)z^{-(n+1)},
 $$
$$
\widehat {\mathbf B}_{\hat{\theta}(t)}(z^{-1}):=\hat b_{1}(t)z^{-1}+
\hat b_{2}(t)z^{-2}\cdots+\hat b_{n}(t)z^{-n}.$$
Now, 
let $\mathbf A^* ( z^{-1} )$ be a $(2 n+1)\text{th}$-order monic polynomial
of the form
\[
\mathbf A^* (z^{-1} ) = 1 + a_1^* z^{-1} + \cdots + a_{2n+1}^* z^{-(2n+1)}  
\]
chosen as desired 
such that $z^{2n+1} \mathbf A^* (z^{-1} )$ has all of its roots inside the unit desk.
Next we design a $(n+1){\text{th}}$-order strictly proper controller: 
at every $t\geq t_0 $
choose the following polynomials
$$\mathbf L_{\hat{\theta}(t)}(z^{-1})=1+l_{1}(t)z^{-1}+
l_{2}(t)z^{-2}+\cdots+l_{n}(t)z^{-n}, \qquad
 $$
$$ 
\mathbf P_{\hat{\theta}(t)}(z^{-1})=p_{1}(t)z^{-1}+
p_{2}(t)z^{-2}+\cdots+p_{n+1}(t)z^{-(n+1)}$$
so they satisfy the equation
\begin{equation}
\label{char1}
\widehat {\overline {\mathbf A}}_{\hat{\theta}(t)}(z^{-1})
\mathbf L_{\hat{\theta}(t)}(z^{-1})
+
\widehat {\mathbf B}_{\hat{\theta}(t)}(z^{-1})
\mathbf P_{\hat{\theta}(t)}(z^{-1})
=\mathbf A^* ( z^{-1} ).
\end{equation}
Since $z^{n+1}\widehat {\overline {\mathbf A}}_{\hat{\theta}(t)}(z^{-1})$ and $z^n\widehat {\mathbf B}_{\hat{\theta}(t)}(z^{-1})$ are coprime, 
we know that there exist {unique} $\mathbf L_{\hat{\theta}(t)}(z^{-1})$ and $\mathbf P_{\hat{\theta}(t)}(z^{-1})$ that satisfy \eqref{char1}---see \cite[Theorem 2.3.1]{Ioannou:1995:RAC:211527}; this entails
solving a linear equation.
It is also easy to prove that the coefficients of $\mathbf L_{\hat{\theta}(t)}(z^{-1})$ and $\mathbf P_{\hat{\theta}(t)}(z^{-1})$ are analytic functions of $\hat{\theta}(t) \in \overline{\cal S}$. We can now define the control law by
\begin{equation}
  \mathbf L_{\hat{\theta}(t-1)}(z^{-1}) \overline U(z)=-\mathbf P_{\hat{\theta}(t-1)}(z^{-1}) \overline Y(z).
  \label{controlTF1}
\end{equation}
This can be written in terms of the state vector. To proceed, define the control gains $K_{\hat{\theta}(t)} \in\R^{2n+1}$ by
\begin{equation}
\label{controlM_para3}
K_{\hat{\theta}(t)}:=
\bigl[
-p_{1}(t) \;\; \cdots \;\; -p_{n+1}(t) \;\; 
 -l_{1}(t) \;\; \cdots \;\; -l_{n}(t)
 \bigr],
\end{equation}
so from \eqref{controlTF1},
we see that 
\begin{equation}
 \overline u(t) = K_{\hat{\theta}(t-1)} \psi(t-1).
 \label{control_ax}
\end{equation}
From the definition of $\overline u(t) $ 
the plant control input
for $t> t_0 $
becomes
\begin{equation}
  u(t) = u(t-1) + \overline u(t).
  \label{control4}
\end{equation}
Next, we provide preliminary setup needed before analyzing the behavior of the closed-loop system.

\section{Preliminary Analysis}
\unskip

We can write down a state-space model of our closed-loop system with
$\psi (t) \in \R^{2n+1}$ as the state, including present and past values of the tracking error and the auxiliary control input. 
From \eqref{predict1}, we have 
\begin{equation}
  \overline y(t+1) =  \hat\theta(t)^\top \psi(t) + e(t+1);
  \label{y_eq2}
\end{equation}
we also have from \eqref{control_ax}
\begin{equation}
 \overline u(t+1) = K_{\hat{\theta}(t)}\psi(t).
 \label{u_eq2}
 \end{equation}
Now, define for each $\hat\theta(t)$ the matrix ${{\cal A}_{\hat\theta(t) }}\in\R^{(2 n +1)\times (2 n+1)}$
\begin{equation} \label{A_mat1}
{
  { {\cal A}_{\hat\theta(t)}}:=\begin{bmatrix}  
\hat\theta(t)^\top
 \\
\begin{bmatrix}
I_{n} & & {\mathbf 0}_{n\times(n+1)}
\end{bmatrix}
\\
 K_{\hat{\theta}(t)}
 \\
 \begin{bmatrix}
{\mathbf 0}_{(n-1)\times(n+1)}
&
I_{n-1} & {\mathbf 0}_{(n-1)\times 1}
\end{bmatrix}
\end{bmatrix};
}
\end{equation}
using the definition \eqref{A_mat1},
from \eqref{y_eq2} and \eqref{u_eq2}
we see that
the following equation holds:
\begin{equation}
\psi (t+1) = {\cal A}_{\hat{\theta}(t)} \psi (t) +
\e_1 e(t+1). 
\label{keyeq2}
\end{equation}
Notice that for every frozen time $t$ the characteristic equation of ${\cal A}_{\hat\theta(t)}$
satisfies
\begin{flalign*}
&
\det \bigl(
z I_{2n +1} - {\cal A}_{\hat\theta(t)}
\bigr)
\\
&
\quad=
z^{2 n+1}[
\widehat {\overline {\mathbf A}}_{\hat{\theta}(t)}(z^{-1})
\mathbf L_{\hat{\theta}(t)}(z^{-1})
+
\widehat {\mathbf B}_{\hat{\theta}(t)}(z^{-1})
\mathbf P_{\hat{\theta}(t)}(z^{-1})  
]
\\
&
\quad=
z^{2 n+1}\mathbf A^* ( z^{-1} ).
\end{flalign*}
Now define 
\begin{equation}
  \Xi(t)
:=
\e_1 \frac{ e(t+1)}{\|\psi(t)\|^2}\psi(t)^\top
\label{delta_2}
\end{equation}
where it is easy to see that
$
  \e_1 e(t+1) = \Xi(t)\psi(t) 
    ,$
so
we also obtain from \eqref{keyeq2}:
\begin{equation}
\psi (t+1) = [{\cal A}_{\hat{\theta}(t)} + \Xi(t) ] \psi (t). 
\label{keyeq}
\end{equation}
While we can view \eqref{keyeq} as a linear time-varying system, 
we have to keep in mind that ${\cal A}_{\hat{\theta}(t)}$ and $\Xi(t) $ are implicit nonlinear functions
of $\theta $, $\theta_0 $, $\boldsymbol\phi_0$, $r $ and $w $.
However, this linear time-varying interpretation is convenient for analysis.

Before presenting the main result of this paper,
in the form of a Proposition,
 we first analyze the adaptive control system to obtain
 a desired bound on $\psi(t) $.

Before proceeding, define
\[
\underline\lambda:=
\max
\bigl\{|\lambda|: \lambda\in\C \text{ and }  \lambda^{2n+1}{\mathbf A}^*(\lambda^{-1})=0  \bigr\}.
\]

\begin{prop}
\label{lemma_main} 
Suppose that the adaptive controller \eqref{est2}, 
\eqref{controlM_para3}, \eqref{control_ax} and \eqref{control4}
is applied to the plant \eqref{plant1}.
Then for every 
$\lambda \in (\underline\lambda,1)$, 
there exists a constant $c>0$ so that 
for every 
$t_0\in\Z$, 
$\theta \in {\cal S} $,
$\boldsymbol\phi_0 \in {\R}^{2(n+1)}$,
$\theta_0\in\overline{\cal S}$, 
$\mu>0 $, $r\in\R $
and
$w \in {\ellb}_{\infty}$,
and for every interval of the form $[\underline t,\overline t ]
 \subset [t_0,\infty ) $ 
which satisfies 
$\|\psi(t)\|^2 \geq \mu, 
\,
t\in[\underline t,\overline t ),$
the following holds
\begin{equation} 
\| \psi (t) \| 
\leq
 c  \lambda^{t-\underline t} \|\psi(\underline t) \|
+
\sum_{j=\underline t}^{t-1} 
c\lambda^{t-j-1} |\overline w(j)|, 
\quad
t \in [\underline t,\overline t].
\label{prop_psi}
\end{equation}
\end{prop}

\begin{proof}
  See the Appendix.
The proof there utilizes Proposition \ref{est_prop} and a technical result of Kreisselmeier's \cite{kreiss2}.
\end{proof}

\section{The Main Result}

We now present the main result of this paper.
We will define a vector
$\phi(t)\in\R^{2(n+1) } $
\begin{flalign*}
  \phi(t)
  &:=
\left[ 
\begin{matrix}
y(t) & y(t-1) & \cdots & y(t- n)
\end{matrix}\right.
\\
&\qquad\qquad\quad
\left.
\begin{matrix} 
& u(t) & u(t-1) & \cdots & u(t-n)
\end{matrix}
\right]
^\top
\end{flalign*}
to serve as the overall system's state;
note here that the vector $\psi $ contains values of the tracking error and the auxiliary control input, 
while the vector $\phi $ contains values 
of the plant input and output.

\begin{theorem}
\label{theorem_main}
Suppose that the adaptive controller 
\eqref{est2}, 
\eqref{controlM_para3}, \eqref{control_ax} and \eqref{control4}
is applied to the plant \eqref{plant1}.
Then for every 
$\lambda \in (\underline\lambda,1)$, 
there exists a constant $\gamma>0$ so that 
for every 
$t_0\in\Z$, 
$\theta \in {\cal S} $,
$\boldsymbol\phi_0 \in {\R}^{2(n+1)}$,
$\theta_0\in\overline{\cal S}$, 
$\mu>0 $, $r\in\R $
and $w \in {\ellb}_{\infty}$,
\\
\indent
i) the following bound holds:
\begin{flalign}
\| \phi (t) \| 
&
\leq
 \gamma  \lambda^{t-t_0 } \|\boldsymbol\phi_0 \|
 + 
\gamma(|r|+\sqrt{\mu})
 +
 \nonumber
\\
&
\qquad\qquad
\gamma\sum_{j=t_0}^{t-1}
\lambda^{t-1-j} |w(j)| 
, 
\quad 
t \geq t_0;
 \label{thm_bound1}
\end{flalign}
\indent
ii) if $w$ is constant, then
$y(t)\rightarrow r  $ as $t \rightarrow \infty $.
\end{theorem}

\begin{proof}
 See the Appendix. 
The proof there
utilizes
Proposition \ref{lemma_main} to analyze
 the closed-loop behavior on time intervals 
 where $\psi(\cdot)$ is {\it large}, i.e. 
 when $\|\psi(t)\|^2 \geq \mu$, while it
 analyzes the closed-loop behavior where $\psi(\cdot)$ is {\it small} in a direct manner, before combining both cases. Linear system theory is then used to translate the bound on $\psi $ 
 to a bound on $\phi $. 
\end{proof}

\begin{remark}
The above result shows that the closed-loop 
system experiences linear-like behavior. 
There is a {\bf uniform} exponential decay bound on the effect of the initial condition, and a convolution bound on the effect of the noise.
It also provides a
clear bound in terms of the parameter $\mu$ and the set-point $r$---the smaller
that they are, the smaller that this bound is.
\end{remark}
\begin{remark}
The bound in \eqref{thm_bound1} implies that the system has a bounded gain
with bias (from $w$ to $y$):
\begin{flalign*}
  \Vert \phi(t) \Vert 
    &\leq
 \tfrac{\gamma }{1-\lambda} 
 \left(
   \lambda^{t-t_0 }
 \Vert \boldsymbol\phi_0 \Vert
   +
    \|w\|_{\infty}+|r| + \sqrt{\mu}
    \right),
  \; t\geq t_0.
\end{flalign*}
\end{remark}

\begin{remark}
We can show that the convolution bound proven in Theorem
\ref{theorem_main} will guarantee robustness to a degree of time-variations
and unmodelled dynamics. Furthermore,
we can also obtain bounds on the average tracking
error both in the case of no noise under slow time-variations,
as well as in the noisy case. We obtain these results by applying similar 
arguments to those in the proofs of Theorems 2, 3 and 4 of \cite{tac25} which considered the MRAC problem.
\end{remark}

\section{A Simulation Example}
We provide here a simulation example to illustrate the results of this paper. Consider the second-order plant:
\begin{gather*}
y(t+1)
=
a_1y(t)+a_2y(t-1)+
b_1u(t)+b_2u(t-1)+w(t)
\end{gather*}
with parameters belonging to the uncertainty set ${\cal S}$:
\begin{gather*}
{\cal S}:=\biggl\{\left[ 
\begin{matrix}
a_1 & a_2 & b_1 & b_2
\end{matrix} \right]^\top\in{\R}^{4}:
\\
a_1\in[-2,0],a_2\in[-3,-1],b_1\in[-1,0],b_2\in[-5,-3]
\biggr\}.
\end{gather*}
Hence, every admissible plant model is unstable and non-minimum phase, which makes this plant challenging to control; it has two complex unstable poles together with a zero that can lie in $[3,\infty)$. It is also obvious that ${\cal S}$ is a convex set. We define the set $\overline{\cal S}$ by \eqref{const_multi}:
\begin{gather*}
\overline{\cal S}:=\biggl\{\left[ 
\begin{matrix}
\overline a_1 & \overline a_2 & \overline a_3 & b_1 & b_2
\end{matrix} \right]^\top\in{\R}^{5}:
\\
\bar a_1\in[-1,1],\bar a_2\in[-3,1],
\bar a_3\in[1,3],b_1\in[-1,0],b_2\in[-5,-3]
\biggr\},
\end{gather*}
which will be utilized in estimating the parameters of the auxiliary plant
$
\theta^*
\in \overline{\cal S}.
 $
It is clear that the set $\overline{\cal S}$ is also compact and convex, and satisfies the coprimeness requirement.

For this simulation we set $a_1 = -\tfrac{1}{2}$, $a_2 = -\tfrac{3}{2}$, $b_1 = -\tfrac{3}{4}$, and $b_2= -3$. We will apply the proposed controller \eqref{est2}, 
\eqref{controlM_para3}, \eqref{control_ax} and \eqref{control4}
where we choose $\mathbf A^*(z^{-1}) = 1- \tfrac{3}{5}z^{-1} $ where clearly the roots of $z^{2n+1}\mathbf A^*(z^{-1})= z^5 - \tfrac{3}{5}z^{4} $ are inside the unit desk as required.
We set the desired reference signal to be of a constant magnitude: $|y^*(t)|=2$ but with its sign changing every $200$ steps.
Set initial conditions to $y(0)=y(-1)=y(-2)=-1$ and $u(-1)=u(-2)=0$; we also set
$\theta_0=
\begin{bmatrix} 0 & -1 & 2 & -\tfrac{1}{2}& -4
\end{bmatrix}
^\top$.
We also add a disturbance of constant magnitude: $|w(t)|=\tfrac{1}{2}$, but with its sign changing every $250$ steps.
Figures \ref{figex1a} and \ref{figex1b} displays the results. 
We see that the controller does a good job of tracking; the closed-loop system experiences some transient behavior when the 
reference signal changes or when the
disturbance changes, but the tracking recovers quickly.

 \begin{figure}
\center\includegraphics[trim=25 10 25 10,clip,width=1\columnwidth]{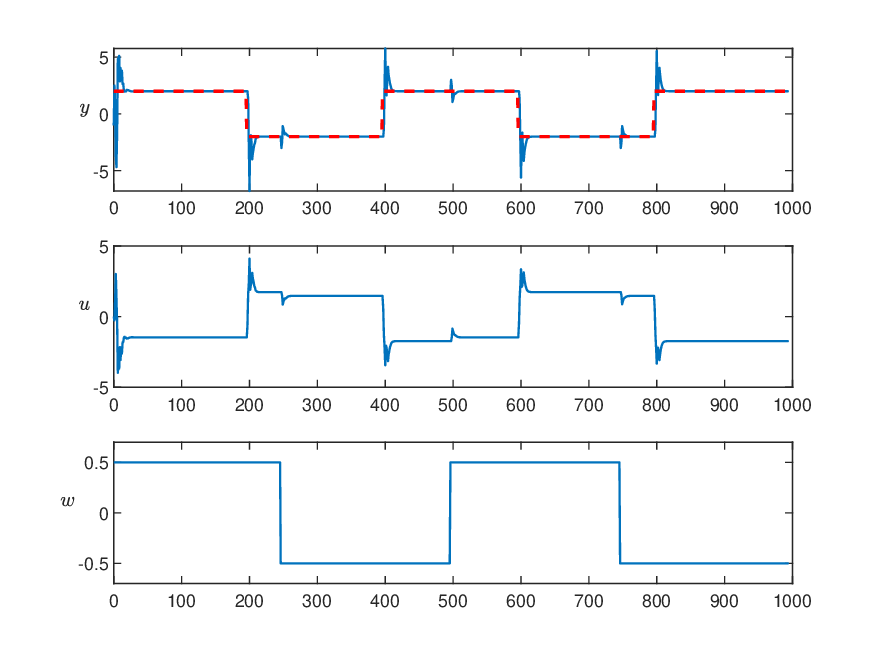}
\vspace{-1.5em}
\caption{The top plot shows the output and the reference signals;
the middle plot shows the control input;
the bottom plot shows the disturbance signal.}
\vspace{-1em}
\hfill
\label{figex1a}
\end{figure}

 \begin{figure}
\center\includegraphics[trim=25 10 25 10,clip,width=1\columnwidth]{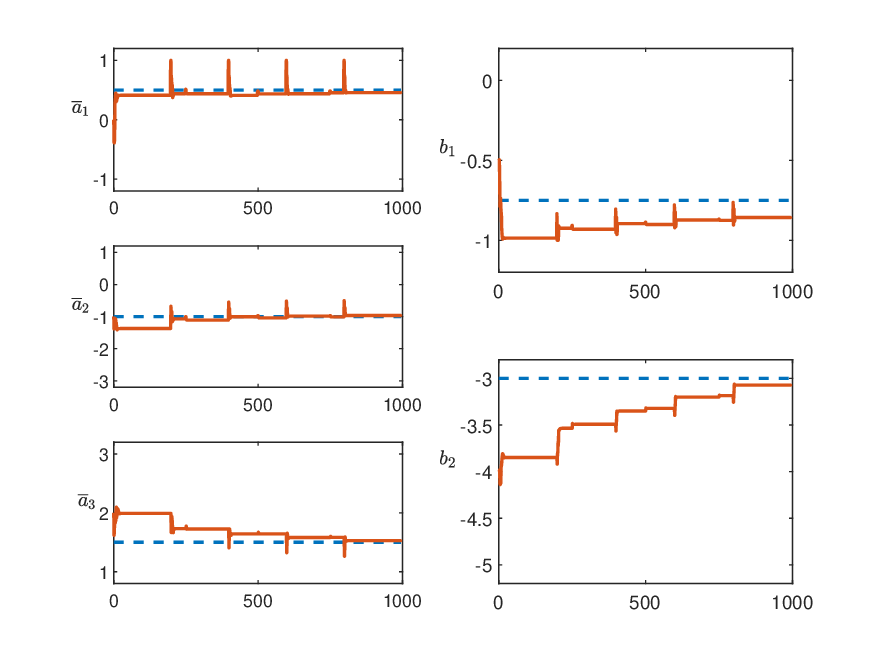}
\vspace{-1.5em}
\caption{The plots show the estimated (solid) and the actual (dashed) parameters.}
\vspace{-1em}
\hfill
\label{figex1b}
\end{figure}

\section{Summary and Conclusions}

In this paper 
we examine 
the set-point tracking
problem when the 
commonly used
projection
algorithm is utilized, subject to several common assumptions on the
set of admissible parameters.
While in the literature it is proven that the nonlinear closed-loop
system is bounded-input bounded-state, here we prove 
a linear-like bound on the closed-loop behavior
consisting of
1)
a decaying exponential on the initial condition,
2)
a linear-like convolution bound on the noise input, and
3)
a constant 
scaled by the square root of the constant in the denominator
of the estimator update law and the set-point.
Furthermore, in the constant disturbance case, 
asymptotic tracking is proven as well.

We would like to extend these desirable results 
to the 
cases when 
the uncertainty set is non-convex and where its convex hull violate 
the comprimeness requirement and/or the plant order is not exactly known;
an adaptive controller 
using multiple estimators and switching, along the lines of \cite{TAC22},
may be needed.

\section*{Acknowledgments}
The author would like to thank Daniel E. Miller for valuable discussions that contributed to the development of this paper.

\bibliographystyle{IEEEtranSmod_v2}
\balance
\bibliography{cdc25_pp_arxiv_v1.bib}

\begin{thebibliography}{10}
\providecommand{\url}[1]{#1}
\csname url@samestyle\endcsname
\providecommand{\newblock}{\relax}
\providecommand{\bibinfo}[2]{#2}
\providecommand{\BIBentrySTDinterwordspacing}{\spaceskip=0pt\relax}
\providecommand{\BIBentryALTinterwordstretchfactor}{4}
\providecommand{\BIBentryALTinterwordspacing}{\spaceskip=\fontdimen2\font plus
\BIBentryALTinterwordstretchfactor\fontdimen3\font minus
  \fontdimen4\font\relax}
\providecommand{\BIBforeignlanguage}[2]{{%
\expandafter\ifx\csname l@#1\endcsname\relax
\typeout{** WARNING: IEEEtranS.bst: No hyphenation pattern has been}%
\typeout{** loaded for the language `#1'. Using the pattern for}%
\typeout{** the default language instead.}%
\else
\language=\csname l@#1\endcsname
\fi
#2}}
\providecommand{\BIBdecl}{\relax}
\BIBdecl

\bibitem{akhtar}
S.~{Akhtar} and D.~S. {Bernstein}, ``Lyapunov-stable discrete-time model
  reference adaptive control,'' in \emph{Proc. Amer. Control Conf.}, Jun. 2005,
  pp. 3174--3179.

\bibitem{convex_book}
S.~Boyd and L.~Vandenberghe, \emph{Convex Optimization}.\hskip 1em plus 0.5em
  minus 0.4em\relax Cambridge University Press, 2004.

\bibitem{desoer}
C.~A. Desoer, ``{Slowly varying discrete system $x_{i+1}=A_i x_i$},''
  \emph{Electronics Letters}, vol.~6, no.~11, pp. 339--340, May 1970.

\bibitem{morse1978}
A.~Feuer and A.~S. Morse, ``{Adaptive control of single-input, single-output
  linear systems},'' \emph{IEEE Trans. Autom. Control}, vol.~23, no.~4, pp.
  557--569, Aug. 1978.

\bibitem{GOODWIN1987}
G.~C. Goodwin and D.~Q. Mayne, ``A parameter estimation perspective of
  continuous time model reference adaptive control,'' \emph{Automatica},
  vol.~23, no.~1, pp. 57--70, 1987.

\bibitem{Goodwin1980}
G.~C. Goodwin, P.~Ramadge, and P.~Caines, ``{Discrete-time multivariable
  adaptive control},'' \emph{IEEE Trans. Autom. Control}, vol.~25, no.~3, pp.
  449--456, Jun. 1980.

\bibitem{goodwinsin}
G.~C. Goodwin and K.~S. Sin, \emph{{Adaptive Filtering Prediction and
  Control}}.\hskip 1em plus 0.5em minus 0.4em\relax New York, NY, USA: Dover
  Publications, Inc., 1984.

\bibitem{Ioa86}
P.~A. Ioannou and K.~S. Tsakalis, ``{A robust direct adaptive controller},''
  \emph{IEEE Trans. Autom. Control}, vol.~31, no.~11, pp. 1033--1043, Nov.
  1986.

\bibitem{Ioannou:1995:RAC:211527}
P.~A. Ioannou and J.~Sun, \emph{Robust Adaptive Control}.\hskip 1em plus 0.5em
  minus 0.4em\relax Upper Saddle River, NJ, USA: Prentice-Hall, Inc., 1995.

\bibitem{kreiss}
G.~Kreisselmeier and B.~D.~O. Anderson, ``{Robust model reference adaptive
  control},'' \emph{IEEE Trans. Autom. Control}, vol.~31, no.~2, pp. 127--133,
  Feb. 1986.

\bibitem{kreiss2}
G.~Kreisselmeier, ``{Adaptive control of a class of slowly time-varying
  plants},'' \emph{Syst. Control Lett.}, vol.~8, no.~2, pp. 97--103, Dec. 1986.

\bibitem{hanfu}
Y.~Li and H.-F. Chen, ``{Robust adaptive pole placement for linear time-varying
  systems},'' \emph{IEEE Trans. Autom. Control}, vol.~41, no.~5, pp. 714--719,
  May 1996.

\bibitem{rick}
R.~H. Middleton and G.~C. Goodwin, ``{Adaptive control of time-varying linear
  systems},'' \emph{IEEE Trans. Autom. Control}, vol.~33, no.~2, pp. 150--155,
  1988.

\bibitem{rick2}
R.~H. Middleton, G.~C. Goodwin, D.~J. Hill, and D.~Q. Mayne, ``{Design issues
  in adaptive control},'' \emph{IEEE Trans. Autom. Control}, vol.~33, no.~1,
  pp. 50--58, Jan. 1988.

\bibitem{scl17}
D.~E. Miller, ``{A parameter adaptive controller which provides exponential
  stability: The first order case},'' \emph{Syst. Control Lett.}, vol. 103, pp.
  23--31, May 2017.

\bibitem{ccta17}
D.~E. Miller, ``{Classical discrete-time adaptive control revisited:
  Exponential stabilization},'' in \emph{Proc. IEEE Conf. Control Technol.
  Appl.}, Aug. 2017, pp. 1975--1980.

\bibitem{mcss18}
D.~E. Miller and M.~T. Shahab, ``Classical pole placement adaptive control
  revisited: linear-like convolution bounds and exponential stability,''
  \emph{Math. Control Signals Syst.}, vol.~30, no.~4, p.~19, Nov. 2018.

\bibitem{acc19}
D.~E. {Miller} and M.~T. {Shahab}, ``Classical $d$-step-ahead adaptive control
  revisited: Linear-like convolution bounds and exponential stability,'' in
  \emph{Proc. Amer. Control Conf.}, Jul. 2019, pp. 417--422.

\bibitem{mcss20}
D.~E. {Miller} and M.~T. {Shahab}, ``Adaptive tracking with exponential
  stability and convolution bounds using vigilant estimation,'' \emph{Math.
  Control Signals Syst.}, vol.~32, pp. 241--291, 2020.

\bibitem{Morse1980}
A.~S. Morse, ``{Global stability of parameter-adaptive control systems},''
  \emph{IEEE Trans. Autom. Control}, vol.~25, no.~3, pp. 433--439, Jun. 1980.

\bibitem{hyst92}
A.~S. Morse, D.~Q. Mayne, and G.~C. Goodwin, ``Applications of hysteresis
  switching in parameter adaptive control,'' \emph{IEEE Trans. Autom. Control},
  vol.~37, no.~9, pp. 1343--1354, 1992.

\bibitem{Naik}
S.~M. Naik, P.~R. Kumar, and B.~E. Ydstie, ``{Robust continuous-time adaptive
  control by parameter projection},'' \emph{IEEE Trans. Autom. Control},
  vol.~37, no.~2, pp. 182--197, 1992.

\bibitem{Narendra1980}
K.~S. Narendra and Y.-H. Lin, ``{Stable discrete adaptive control},''
  \emph{IEEE Trans. Autom. Control}, vol.~25, no.~3, pp. 456--461, Jun. 1980.

\bibitem{Narendra1980_pt2}
K.~S. Narendra, Y.-H. Lin, and L.~Valavani, ``{Stable adaptive controller
  design, part II: Proof of stability},'' \emph{IEEE Trans. Autom. Control},
  vol.~25, no.~3, pp. 440--448, Jun. 1980.

\bibitem{rohrs}
C.~Rohrs, L.~Valavani, M.~Athans, and G.~Stein, ``{Robustness of
  continuous-time adaptive control algorithms in the presence of unmodeled
  dynamics},'' \emph{IEEE Trans. Autom. Control}, vol.~30, no.~9, pp. 881--889,
  Sep. 1985.

\bibitem{cdc18}
M.~T. {Shahab} and D.~E. {Miller}, ``Multi-estimator based adaptive control
  which provides exponential stability: The first-order case,'' in \emph{Proc.
  IEEE Conf. Decis. Control}, Dec. 2018, pp. 2223--2228.

\bibitem{cdc19}
M.~T. Shahab and D.~E. Miller, ``{Adaptive Set-Point Regulation using Multiple
  Estimators},'' in \emph{Proc. IEEE Conf. Decis. Control}, Dec. 2019, pp.
  84--89.

\bibitem{cdc21}
M.~T. Shahab and D.~E. Miller, ``{Model Reference Adaptive Control with
  Linear-like Closed-loop Behavior},'' in \emph{Proc. IEEE Conf. Decis.
  Control}, Dec. 2021, pp. 1069--1074.

\bibitem{TAC22}
M.~T. Shahab and D.~E. Miller, ``Asymptotic tracking and linear-like behavior
  using multi-model adaptive control,'' \emph{IEEE Trans. Autom. Control},
  vol.~67, no.~1, pp. 203--219, 2022.

\bibitem{acs24}
M.~T. Shahab and D.~E. Miller, ``Inherent robustness in the adaptive control of
  a large class of systems,'' \emph{International Journal of Adaptive Control
  and Signal Processing}, vol.~38, no.~7, pp. 2423--2442, 2024.

\bibitem{tac25}
M.~T. Shahab and D.~E. Miller, ``Revisiting model reference adaptive control:
  Linear-like closed-loop behavior,'' \emph{IEEE Trans. Autom. Control},
  vol.~70, no.~3, pp. 1483--1498, 2025.

\bibitem{Tsakalis4}
K.~S. Tsakalis and P.~A. Ioannou, ``{Adaptive control of linear time-varying
  plants: a new model reference controller structure},'' \emph{IEEE Trans.
  Autom. Control}, vol.~34, no.~10, pp. 1038--1046, 1989.

\bibitem{Wen}
C.~Wen, ``{A robust adaptive controller with minimal modifications for discrete
  time-varying systems},'' \emph{IEEE Trans. Autom. Control}, vol.~39, no.~5,
  pp. 987--991, May 1994.

\bibitem{Wenhill}
C.~Wen and D.~J. Hill, ``{Global boundedness of discrete-time adaptive control
  just using estimator projection},'' \emph{Automatica}, vol.~28, no.~6, pp.
  1143--1157, Nov. 1992.

\bibitem{ydstie}
B.~E. Ydstie, ``{Transient performance and robustness of direct adaptive
  control},'' \emph{IEEE Trans. Autom. Control}, vol.~37, no.~8, pp.
  1091--1105, 1992.

\end{thebibliography}



\pagebreak

\appendix

Before proceeding, define
\[
\overline \alpha := \max_{\theta^* \in \overline{\cal S}} \| {\cal A}_{\theta^*} \|,
\qquad 
\overline{\mathbf s} := \max_{\theta_1, \theta_2 \in \overline{\cal S}} \|\theta_1 - \theta_2 \|,
\]
and 
let $\lceil\cdot\rceil$ denote the ceiling function.
We want also to find a crude bound on the closed-loop behavior. 
From the definitions of the auxiliary plant \eqref{plant2}
and the prediction error \eqref{predict1}, it is easy to see that
\begin{flalign}
e(t+1) &= - \tilde \theta (t)^\top \psi(t) + \overline w(t)
\label{pred_tilde1}
\\
\Rightarrow
|e(t+1)| &\leq \overline{\mathbf s}\|\psi(t) \| + |\overline w (t) |,
\quad t\geq t_0.
\end{flalign}
From this and \eqref{keyeq2}, we can obtain the following crude bound on $\psi(\cdot) $:
\begin{equation}
  \|\psi(t+1) \| 
  \leq 
  (\overline\alpha + \overline{\mathbf s})\|\psi(t) \| + |\overline w (t) |,
  \quad t\geq t_0.
  \label{crude2}
\end{equation}

The following result of Kreisselmeier's is useful in analyzing our closed-loop system.

\begin{prop}[\hspace{-.2pt}\textbf{\cite{kreiss2}}]
\label{prop2}
Consider the discrete-time system
\[
x (t+1) = [ A (t) + \Delta (t) ] x(t),
\]
and let $\Phi_{A+ \Delta} (t, \tau  ) $ to denote its state transition matrix.
Suppose that there exist constants
$\sigma \in (0,1)$, $\gamma_1 >1$, $\alpha_i \geq 0$, 
and $\beta_i \geq 0$ ($i=0,1,2$)
so that

\begin{description}
\item{
  (i)
  }
for all $t \geq t_0 $, we have
$\| A (t)^k \| \leq \gamma_1 \sigma^k , \; k \geq 0 $;

\item {
  (ii)
  }
for all $t > \tau\geq t_0 $ we have
\[
\sum_{k = \tau}^{t-1} \| A (k+1) - A (k) \| \leq
\alpha_0 + \alpha_1 ( t- \tau )^{1/2} + \alpha_2 ( t- \tau ),
\]
and
$$\sum_{k = \tau}^{t-1} \|  \Delta (k) \| \leq
\beta_0 + \beta_1 ( t- \tau )^{1/2} + \beta_2 ( t- \tau ); $$

\item{
  (iii)
}
there exists a $\rho \in ( \sigma , 1 )$ and $N \in \N$ satisfying
{
$$ \alpha_2 + \frac{\beta_2}{N} < \frac{\left( \frac{\rho}{\sqrt[N]{\gamma_1}} -
\sigma \right)}{N \gamma_1}.  $$
}
\end{description}
\unskip\noindent
Then
there exists a constant $\overline\gamma$ so that 
\[
\| \Phi_{A+ \Delta} (t, \tau  ) \| \leq \overline\gamma \rho^{t - \tau }  , \; t \geq \tau,
\]
where 
\begin{equation}
\overline\gamma = 
\left[
  \gamma_1 + \sum_{i=0}^2\beta_i
\right]^N 
\sup_{j\geq 0}
\left(
  \frac{\sqrt[N]{\gamma_1}}{\rho} \left[
    \sigma + \gamma_1 \sum_{i=0}^2 (N\alpha_i+\beta_i)(jN)^{\frac{i-2}{2}}
  \right]
\right)^{jN}.
\nonumber
\end{equation}
\end{prop}

\bigskip

\section*{Proof of Proposition \ref{lemma_main}}

\begin{proof}
Fix $\lambda \in ( \underline{\lambda}, 1)$. 
Let $t_0 \in \Z$, 
$\theta_0 \in \overline{\cal S}$,
$\theta \in {\cal S}$,
$\boldsymbol\phi_0 \in \R^{2(n+1)}$,
$\mu>0 $,
$r\in\R $
and 
$w\in \ellb_\infty$ be arbitrary;
as well,
let
$[\underline t,\overline t ] \subset [t_0,\infty ) $ 
be an arbitrary interval
which satisfies 
$\|\psi(t)\|^2 \geq \mu, 
\,
t\in[\underline t,\overline t).$
Now choose
$\lambda_1 \in 
( \underline{\lambda} ,  \lambda )$.

We are going to utilize Proposition \ref{prop2} in analyzing the closed-loop system.
As the characteristic polynomial of
${\cal A}_{\hat\theta(t)}$ is exactly $z^{2n+1} A^* (z^{-1})$ 
for every $t \geq t_0$,
and as
the coefficients of
$\mathbf L_{\hat{\theta}(t)} (z^{-1} )$ and $\mathbf P_{\hat{\theta}(t)} (z^{-1} )$ are the solution of a linear
equation and are analytic functions of
$\hat{\theta} (t) \in \overline{\cal S}$,
then there exists a constant $\overline\alpha$ so that, for all initial
conditions, $w \in \ellb_\infty$ and $r \in \R$, we have
$\sup_{t \geq t_0} \| {\cal A}_{\hat\theta(t)} \| \leq \overline\alpha $.
Then we apply the argument from \cite{desoer} (which considers a more general time-varying situation) to show that 
there exists a $\overline c>0$ so that for every $t \geq t_0$ we have
\begin{equation}
\| {\cal A}_{\hat\theta(t)} ^k \| \leq \overline c \lambda_1^{k} , \;\; k \geq 0. 
\label{cond1}
\end{equation}

We now provide a number of useful bounds.
Let us obtain a bound on $\Xi(\cdot) $ defined in \eqref{delta_2};
by
the Cauchy–Schwarz
inequality, we see that 
\begin{flalign}
\sum_{j=\tau}^{t-1 }\|\Xi(j) \|  
&\leq
\sum_{j=\tau}^{t-1 }\frac{|e(j+1)|}{\|\psi(j) \|}  
\nonumber 
\\
&\leq 
\left[\sum_{j=\tau}^{t-1 }\frac{|e(j+1)|^2}{\|\psi(j) \|^2} \right]^{\frac{1}{2}} 
(t-\tau)^{\frac{1}{2}},
\nonumber 
\\
&
\qquad \quad 
\overline t \geq t > \tau \geq \underline t.
\label{delta_3}
\end{flalign}
Next, by applying Proposition \ref{est_prop}(ii) and
as
the coefficients of
$\mathbf L_{\hat\theta(t)} (z^{-1} )$ and $\mathbf P_{\hat\theta(t)} (z^{-1} )$ are analytic functions of
$\hat{\theta} (t) \in \overline{\cal S}$, 
there exists $\overline K >0$ so 
we can obtain 
\begin{flalign}
\sum_{j=\tau}^{t-1 }\| {\cal A}_{\hat\theta(j+1)} &- {\cal A}_{\hat\theta(j)}   \|  
\leq
\sum_{j=\tau}^{t-1 }
\| \hat\theta(j+1) - \hat\theta(j)   \|  + 
\nonumber 
\\ 
&\qquad
\| K_{\hat\theta(j+1)} - K_{\hat\theta(j)}   \|
\nonumber 
\\
&\leq
\sum_{j=\tau}^{t-1 }
(1 + \overline K)\| \hat\theta(j+1) - \hat\theta(j)   \|
\nonumber 
\\
&\leq
\sum_{j=\tau}^{t-1 }
(1 + \overline K) \frac{|e(j+1)|}{\|\psi(j) \|}  
\nonumber 
\\
&\leq 
(1 + \overline K) 
\left[\sum_{j=\tau}^{t-1 }\frac{|e(j+1)|^2}{\|\psi(j) \|^2} \right]^{\frac{1}{2}} 
(t-\tau)^{\frac{1}{2}},
\nonumber 
\\
&
\qquad \quad
\overline t \geq t > \tau \geq \underline t.
\label{Ak_2}
\end{flalign}

Now we consider the closed-loop system behavior 
on $[\underline t, \overline t ] $.
To proceed, 
we partition the
interval into two parts: one in which the noise $\overline w(\cdot) $ is small versus 
$\psi(\cdot) $
and one where it is not.
To this end, 
with $\mathfrak{v}>0 $ to be chosen shortly, 
define
\begin{equation}
S_{\text{\sf good}}:=
\left\{
j\in[\underline t,\overline t ] : 
\tfrac{|\overline w(j) |^2 }{\|\psi(j) \|^2 } 
< \mathfrak{v}
\right\},
S_{\text{\sf bad}}:=
\left\{
j\in[\underline t,\overline t ] : 
\tfrac{|\overline w(j) |^2 }{\|\psi(j) \|^2 } \geq \mathfrak{v}
\right\};
\nonumber
\end{equation}
clearly $[\underline t,\overline t ] = S_{\text{\sf good}} \cup S_{\text{\sf bad}}$.
Observe that this partition implicitly 
depends on $\theta\in{\cal S} $, as well as the initial
conditions. We will easily obtain bounds on the closed-loop system behavior on $S_{\text{\sf bad}}$; 
we will apply Proposition \ref{prop2} 
to \eqref{keyeq}
to analyze the behavior on $S_{\text{\sf good}}$.
Before proceeding, we partition the timeline
into intervals which oscillate between $S_{\text{\sf good}}$
and $S_{\text{\sf bad}}$. To this end, it is easy to see
that we can define a sequence of intervals of the form $[k_i,k_{i+1}) $ satisfying:
(i) $k_0=\underline t $; (ii) $[k_i,k_{i+1}) $ either belongs to
$S_{\text{\sf good}}$ or $S_{\text{\sf bad}}$; and
(iii) if $k_{i+1}\neq \overline t $ and $[k_i,k_{i+1}) $ belongs
to $S_{\text{\sf good}}$ (respectively, $S_{\text{\sf bad}}$), then the interval $[k_{i+1},k_{i+2}) $ must belong to
$S_{\text{\sf bad}}$ (respectively, $S_{\text{\sf good}}$).

Now we analyze the closed-loop behavior on each interval.

\noindent

\underline{Step 1 {\it(The behavior on $S_{\text{\sf bad}}$).}}
Let $[k_i,k_{i+1})\subset S_{\text{\sf bad}} $ be arbitrary. In this case,
$\tfrac{|\overline w(j) |^2}{\|\psi(j) \|^2 } \geq \mathfrak{v}, j\in [k_i,k_{i+1}) $, 
so,
\begin{flalign}
\|\psi(j) \|\leq \tfrac{1}{\sqrt{\mathfrak{v}} } 
|\overline w(j) |,
\quad j\in[k_i,k_{i+1} );
\label{bad_bd1}
\end{flalign}
then from the crude model \eqref{crude2} 
we have
\begin{flalign}
\|\psi(j+1) \|
&
\leq
\left(\tfrac{\overline\alpha+\overline{\mathbf s}}{\sqrt{\mathfrak{v}} }+1\right)
|\overline w(j) | 
\quad
j\in[k_i,k_{i+1});
\nonumber
\end{flalign}
combining this with \eqref{bad_bd1} yields:
\begin{equation}
\|\psi(j) \|
\leq 
\left\lbrace
\begin{matrix*}[l]
\tfrac{1}{\sqrt{\mathfrak{v}} }
|\overline w(j) |,
& j=k_i
\\ 
\left(\tfrac{\overline\alpha+\overline{\mathbf s}}{\sqrt{\mathfrak{v}} }+1\right)
|\overline w(j-1) | 
,
&j=k_i+1,\ldots,k_{i+1}.
\end{matrix*}
\right.
\label{badCase_bd1}
\end{equation}

\underline{Step 2 {\it(The behavior on $S_{\text{\sf good}}$).}}
Suppose that $[k_i,k_{i+1} ) $ lies in $S_{\text{\sf good}}$.
By Proposition \ref{est_prop}(ii) and using the facts 
that $\|\tilde \theta(j) \|\leq \overline{\mathbf s} $
and  
$\frac{|\overline w(j)|^2 }{\|\psi(j) \|^2}< \mathfrak{v}   $ for $j\in[k_i,k_{i+1} ) $,
from \eqref{delta_3}
we obtain:
\begin{flalign}
&\sum_{j=\tau}^{t-1}
\|\Xi(j)\|
\leq 
\left[
4\overline{\mathbf s}^2+
 8
\mathfrak{v}(t-\tau)
\right]^\frac{1}{2}
(t-\tau)^\frac{1}{2} 
\nonumber
\\
&\leq 
2\overline{\mathbf s}(t-\tau)^\frac{1}{2}+
 (8
\mathfrak{v})^\frac{1}{2} (t-\tau)
 ,
\quad \underline k_{i+1} \geq t>\tau \geq \overline k_i.
\label{delta_4}
\end{flalign}
In a similar manner, we can obtain from \eqref{Ak_2}:
\begin{flalign}
&\sum_{j=\tau}^{t-1}
\| {\cal A}_{\hat\theta(j+1)} - {\cal A}_{\hat\theta(j)}   \|
\leq 
2\overline{\mathbf s}(1+\overline K)(t-\tau)^\frac{1}{2}+
\nonumber
\\
&\quad
 (1+\overline K)(8
\mathfrak{v})^\frac{1}{2} (t-\tau)
 ,
\quad \underline k_{i+1} \geq t>\tau \geq \overline k_i.
\label{Ak_4}
\end{flalign}
We should now apply Proposition \ref{prop2} to \eqref{keyeq}: 
set $\alpha_0=\beta_0=0, \alpha_1=2\overline{\mathbf s}(1+\overline K), \alpha_2=(1+\overline K)(8
\mathfrak{v})^\frac{1}{2}, 
\beta_1=2\overline{\mathbf s}, \beta_2=(8
\mathfrak{v})^\frac{1}{2},
\sigma=\lambda_1,\text{ and } \gamma_1=\overline c$;
we see that Conditions (i) and (ii) of Proposition \ref{prop2} are satisfied according to \eqref{cond1}, \eqref{delta_4} and \eqref{Ak_4}. 
Now set $\rho=\lambda $ and choose
$N := \left\lceil \tfrac{\ln \overline c }{\ln\lambda - \ln \lambda_1}
\right\rceil,$
so $\frac{\lambda}{\sqrt[N]{\overline c}}-\lambda_1 >0 $;
we are now able to satisfy Condition (iii) of Proposition \ref{prop2} if we set
$\mathfrak{v}:=
\left(
  \tfrac{
    \frac{\lambda}{\sqrt[N]{\overline c}}-\lambda_1
    }{
4\overline c (N(1+\overline K)+1)
  }
\right)^2.$
Then applying Proposition \ref{prop2} to the discrete-time system \eqref{keyeq},
with $A(t)={\cal A}_{\hat\theta(t)} $ and $\Delta(t)=\Xi(t) $ in this case,
yields that there exists a constant $\overline\gamma $ so that
$\|\boldsymbol\Phi_{A+\Delta} (t,\tau)\|
\leq
\overline\gamma \lambda^{t-\tau},
\,
k_{i+1} \geq t>\tau \geq k_i.$
This means that 
\begin{equation}
  \|\psi(t) \| \leq 
  \overline\gamma \lambda^{t-\tau} \|\psi(\tau) \| ,
\quad  k_{i+1} \geq t>\tau \geq  k_i,
\label{case2_prop}
\end{equation}
which completes Step 2.

\underline{Step 3 {\it(The behavior on $[\underline t,\overline t]$).}}
We now glue together the bounds on $S_{\text{\sf good}}$ and $S_{\text{\sf bad}}$ to obtain a bound which holds on all of $[\underline t,\overline t]$.
If $[k_0,k_1 ) = [\underline t,k_1 ) \subset S_{\text{\sf good}}  $,
then \eqref{prop_psi} holds for $[\underline t,k_1 ]$  
by \eqref{case2_prop} as long as $c \geq \overline\gamma $. 
If $[k_0,k_1 ) = [\underline t,k_1 ) \subset S_{\text{\sf bad}}  $,
then from \eqref{badCase_bd1} 
we see that 
\begin{equation}
\|\psi(j) \|
\leq 
\left\lbrace
\begin{matrix*}[l]
\|\psi(\underline t) \|
,
& j=k_0=\underline t
\\ 
\left(\tfrac{\overline\alpha+\overline{\mathbf s}}{\sqrt{\mathfrak{v}} }+1\right)
|\overline w(j-1) | 
,
&j=k_0+1,\ldots,k_{1},
\end{matrix*}
\right.
\nonumber
\end{equation}
so
\eqref{prop_psi} holds as well as long as 
$c \geq \tfrac{\overline\alpha+\overline{\mathbf s}}{\sqrt{\mathfrak{v}} }+1 $.

We now use induction;
suppose that \eqref{prop_psi} holds for
$[k_0,k_l ] $; we need to prove that 
it holds for $t\in(k_l, k_{l+1} ] $ as well. 
If $[k_l,k_{l+1} ) \subset S_{\text{\sf bad}} $,
then from \eqref{badCase_bd1} we see that 
\eqref{prop_psi} holds on $(k_l, k_{l+1} ]$ also as long as $c \geq \tfrac{\overline\alpha+\overline{\mathbf s}}{\sqrt{\mathfrak{v}} }+1  $. On the other hand, 
if $[k_l,k_{l+1} ) \subset S_{\text{\sf good}}  $,
then $k_l-1 \in S_{\text{\sf bad}}  $, so
from \eqref{badCase_bd1} we have that
$\|\psi(k_l )\|
\leq
\left(
  \tfrac{\overline\alpha+\overline{\mathbf s}}{\sqrt{\mathfrak{v}} }+1
\right)
|\overline w(k_l -1) |;$
combining this with \eqref{case2_prop},
we have 
\begin{flalign}
&\|\psi(t)\| 
\leq 
\overline\gamma \lambda^{t-k_l } \|\psi (k_l ) \|
\nonumber
\\
&
\leq 
\overline\gamma 
\left(
  \tfrac{\overline\alpha+\overline{\mathbf s}}{\sqrt{\mathfrak{v}} }+1
\right)
\lambda^{t-k_l } |\overline w (k_l-1 ) |,
\quad t\in[k_l , k_{l+1}  ].
\end{flalign}
So the bound \eqref{prop_psi}
holds as long as $c \geq \overline\gamma 
\left(
  \tfrac{\overline\alpha+\overline{\mathbf s}}{\sqrt{\mathfrak{v}} }+1
\right) $.
This concludes the proof of Proposition \ref{lemma_main}.
\end{proof}


\section*{Proof of Theorem \ref{theorem_main}}

\begin{proof}
Fix $\lambda \in ( \underline{\lambda}, 1)$. 
Let $t_0 \in \Z$, 
$\theta_0 \in \overline{\cal S}$,
$\theta \in {\cal S}$,
$\boldsymbol\phi_0 \in \R^{2(n+1)}$,
$\mu>0 $,
$r\in\R $
and 
$w\in \ellb_\infty$ be arbitrary.

\underline{Proving Part i) of Theorem \ref{theorem_main}.}
First,
define the times where $\psi(\cdot) $ is large:
\[
S_{\sf L} :=
\{
t\geq t_0 : \|\psi(t)\|^2 \geq \mu
\},
\]
and the times where $\psi(\cdot) $ is small:
\[
S_{\sf S} :=
\{
t\geq t_0 : \|\psi(t)\|^2 < \mu
\}.
\]
Observe that this partition clearly depends 
on $\theta_0 $, $\theta $, $ \boldsymbol\phi_0 $, $r$, and $w$.
We will apply Proposition \ref{lemma_main} to analyze
 the closed-loop behavior on sub-intervals of 
 $S_{\sf L}$;
 we will analyze the closed-loop behavior on $S_{\sf S}$ in a direct manner. 
 Before doing so,
 we 
  we partition the timeline
into intervals which oscillate between $S_{\sf L}$
and $S_{\sf S}$. To this end, it is easy to see
that we can define a (possibly infinite) sequence of intervals of the form $[k_l,k_{l+1}) $ satisfying:
(i) $k_0=t_0 $; (ii) $[k_l,k_{l+1}) $ either belongs to
$S_{\sf L}$ or $S_{\sf S}$; and
(iii) if $k_{l+1}\neq \infty $ and $[k_l,k_{l+1}) $ belongs
to $S_{\sf L}$ (respectively, $S_{\sf S}$), then the interval $[k_{l+1},k_{l+2}) $ must belong to
$S_{\sf S}$ (respectively, $S_{\sf L}$).

Let $[k_l,k_{l+1} ) \subset S_{\sf S} $ be arbitrary; 
then we clearly have
\[
\|\psi(t) \| \leq \sqrt{\mu}, \quad t\in[k_l,k_{l+1} );
\]
by using \eqref{crude2} we have 
\begin{flalign}
\|\psi(k_{l+1}) \| 
&\leq 
(\overline{\mathbf s}+\overline\alpha)
\|\psi(k_{l+1}-1) \| + |\overline w(k_{l+1}-1) | 
\nonumber 
\\
&\leq 
(\overline{\mathbf s}+\overline\alpha)
\sqrt{\mu} + |\overline w(k_{l+1}-1) | .
\end{flalign}
Combining the above 
\begin{equation}
\label{bad_thm1}
\|\psi(t) \| \leq
\begin{cases}
  \sqrt{\mu} & t\in[k_l,k_{l+1} )
  \\
  (\overline{\mathbf s}+\overline\alpha)
\sqrt{\mu} + |\overline w(k_{l+1}-1) |
& t=k_{l+1}.
\end{cases}
\end{equation}

Now let $[k_l,k_{l+1} ) \subset S_{\sf L} $ be arbitrary. By Proposition \ref{lemma_main}
we know that there exists a constant $c$ so that 
\begin{equation} 
\| \psi (t) \| 
\leq
 c  \lambda^{t-k_l} \|\psi(k_l) \|
+
\sum_{j=k_l}^{t-1} 
c\lambda^{t-j-1} |\overline w(j)|, 
\quad  
t \in [k_l,k_{l+1}].
\label{prop_psi_thm}
\end{equation}

\begin{claim}
\label{claim_no1}
There exists a constant $\gamma $
so that the following bound holds:
\begin{equation}
\|\psi(t)\| 
\leq 
\gamma \lambda^{t-t_0 } \|\psi (t_0 ) \|
+
\gamma
\sum_{j=t_0 }^{t-1}  \lambda^{t-j-1 }
|\overline w(j)|
+\gamma\sqrt{\mu}
,
\quad t\geq t_0 .
\label{thmX_claim1}
\end{equation}
\end{claim}

\begin{subproof}[Proof of Claim \ref{claim_no1}]
If $[k_0,k_1 ) = [t_0,k_1 ) \subset S_{\sf L}  $,
then \eqref{thmX_claim1} holds for $[t_0,k_1 ]$  
by \eqref{prop_psi_thm} as long as $\gamma \geq c $. 
If $[k_0,k_1 ) = [t_0,k_1 ) \subset S_{\sf S}  $,
then from \eqref{bad_thm1} we see that 
\eqref{thmX_claim1} also holds for $[t_0,k_1 ]$ as long as 
$\gamma \geq \overline{\mathbf s}+\overline\alpha+1 $.

We now use induction;
suppose that \eqref{thmX_claim1} holds for
$[k_0,k_l ] $; we need to prove that 
it holds for $t\in(k_l, k_{l+1} ] $ as well. 
If $[k_l,k_{l+1} ) \subset S_{\sf S}  $,
then from \eqref{bad_thm1} we see that 
\eqref{thmX_claim1} holds on $(k_l, k_{l+1} ]$ also as long as $\gamma \geq \overline{\mathbf s}+\overline\alpha+1  $. On the other hand, 
if $[k_l,k_{l+1} ) \subset S_{\sf L}  $,
then $k_l-1 \in S_{\sf S}  $; 
from \eqref{bad_thm1} we have that
$\|\psi(k_l )\|
\leq
(\overline{\mathbf s}+\overline\alpha)
\sqrt{\mu} + |\overline w(k_{l}-1) |;$
combining this with \eqref{prop_psi_thm}
yields
\begin{flalign}
&\| \psi (t) \| 
\leq
 c  \lambda^{t-k_l} [
 (\overline{\mathbf s}+\overline\alpha)
\sqrt{\mu} + |\overline w(k_{l}-1) | ]
+
\nonumber 
\\
&\qquad\quad
\sum_{j=k_l}^{t-1} 
c\lambda^{t-j-1} |\overline w(j)|, 
\nonumber
\\
&\leq
\sum_{j=k_l-1}^{t-1} 
c\lambda^{t-j-1} |\overline w(j)|
+
 c(\overline{\mathbf s}+\overline\alpha)
\sqrt{\mu} ,
\quad t\in[k_l , k_{l+1}  ].
\nonumber
\end{flalign}
So the bound \eqref{thmX_claim1}
holds as long as $\bar\gamma \geq  c(\overline{\mathbf s}+\overline\alpha+1)
 $.
\end{subproof}

We now want to convert the bound on $\psi(\cdot) $ in \eqref{thmX_claim1}
to a bound on $\phi(\cdot) $ including values of the plant input and output.
To this end, we construct a state-space model fot the closed-loop system.
By Assumption \ref{assume1}, for every $\theta\in{\cal S}$, 
there exist $A_{\theta}\in\R^{n\times n}$, $B_{\theta}\in\R^{n\times 1}$, $C_{\theta}\in\R^{1\times n}$
and $F_{\theta}\in\R^{n\times 1}$ so that
we obtain an $n{\text{th}} $-order state-space representation of the original plant
\eqref{plant1}
(in controllable canonical form):
\begin{subequations}
\label{state_space1}
\begin{eqnarray}
x(t+1)
&=&
A_{\theta} x(t) + B_{\theta}u(t) + F_{\theta}w(t) \\
y(t)
&=&
C_{\theta}x(t)
\end{eqnarray}
\end{subequations}
where $(A_{\theta},B_{\theta}) $ are controllable
and 
$(C_{\theta},A_{\theta}) $ are observable;
since ${\cal S}$ is compact, the set of all such 
$(A_{\theta},B_{\theta},C_{\theta},F_{\theta})$
is compact as well. 
We can combine this state-space model
with the definition \eqref{ubar_def}
to obtain the $(n+1)\text{th}$-order system
\begin{subequations}
\label{aug_sys1}
 \begin{flalign}
\begin{bmatrix}
x(t+1) \\ u(t+1)
\end{bmatrix}
&=
\underbrace{
\begin{bmatrix}
A_\theta & B_\theta \\
0 & 1
\end{bmatrix}
}_{=:\overline A_\theta}
\underbrace{
\begin{bmatrix}
x(t) \\ u(t)
\end{bmatrix}
}_{=:\overline x(t)}
+
\nonumber 
\\ 
&\qquad 
\begin{bmatrix}
F_\theta
\\ 
0 
\end{bmatrix}
w(t)
+
\begin{bmatrix}
\boldsymbol 0_{n\times 1}
\\ 
1
\end{bmatrix}
\overline u(t+1)
 \\
\overline y(t)
&=
\underbrace{
\begin{bmatrix}
C_\theta & 0
\end{bmatrix}
}_{=:\overline C_\theta}
\begin{bmatrix}
x(t) \\ u(t)
\end{bmatrix}
-
r
.
\end{flalign}
\end{subequations}
Since \eqref{state_space1} is controllable and observable and does not have a zero at $1$ (Assumption \ref{assm2}), 
it follows that $(\overline C_\theta, \overline A_\theta)$ is observable; hence, there exists a unique matrix $ H$ such that the eigenvalues of $\overline A_\theta+  H \overline C_\theta $ are all zero
and it is well-known that $H$ is a continuous function of $\overline A_\theta $ and $\overline C_\theta $. 
Now rewrite \eqref{aug_sys1} as
\begin{flalign*}
\overline x(t+1) 
&=
\left[\overline A_\theta+  H \overline C_\theta\right]
\overline x(t) 
-H 
\overline y (t) 
+ 
\nonumber
\\
& 
\begin{bmatrix}
F_\theta
\\ 
0 
\end{bmatrix}
w(t)
+
\begin{bmatrix}
\boldsymbol 0_{n\times 1}
\\ 
1
\end{bmatrix}
\overline u(t+1)
-Hr
;
\end{flalign*}
noting that $[\overline A_\theta +  H \overline C_\theta]^{j}=0$ for all $ j\geq n+1 $, the solution of the above equation is
\begin{flalign*}
\overline x(t) 
&= 
\sum_{j=1}^{n+1} 
\left[
\overline A_\theta + H \overline C_\theta
\right]^{j-1}
\biggl( 
-H 
\overline y (t-j) 
+ 
\nonumber 
\\
&\quad
\begin{bmatrix}
F_\theta
\\ 
0 
\end{bmatrix}
w(t-j)
+
\begin{bmatrix}
\boldsymbol 0_{n\times 1}
\\ 
1
\end{bmatrix}
\overline u(t-j+1)
-Hr
\biggr),
\nonumber 
\\
&\quad
\qquad t\geq t_0+n+1.
\end{flalign*}
Observe that there exists a matrix $G_\theta $, a combination of
$\overline A_\theta, \overline C_\theta $ and $H$, so that  
\begin{flalign*}
\overline x(t) 
&= 
G_\theta \psi(t-1) + 
\begin{bmatrix}
\boldsymbol 0_{n\times 1}
\\ 
1
\end{bmatrix}
\overline u(t)
+
\nonumber 
\\
&\quad
\sum_{j=1}^{n+1} 
\left[
\overline A_\theta + H \overline C_\theta
\right]^{j-1}
\biggl( 
\begin{bmatrix}
F_\theta
\\ 
0 
\end{bmatrix}
w(t-j)
-Hr
\biggr),
\nonumber 
\\
&\quad
\qquad t\geq t_0+n+1.
\end{flalign*}
By the compactness of ${\cal S}$ and $\overline{\cal S}$,
and since $u(t) $ is an element of $\overline x(t) $,
then there exists a constant $c_1 $ such that 
\begin{flalign*}
 |u(t) |
&\leq  
c_1\|\psi(t-1)\| + 
|\overline u(t)|
+
c_1
|r|
+
c_1
\sum_{j=t-n-1}^{t-1} 
|w(j)|
,
\nonumber 
\\
&\quad
\qquad t\geq t_0+n+1.
\end{flalign*}
We can now use \eqref{thmX_claim1} to have a bound on $\|\psi(t-1) \| $ and on
$|\overline u(t) | $; after utilizing the definition of $\overline w(\cdot) $ and after simplification, we see that there exists $c_2$ 
such 
\begin{flalign*}
 |u(t) |
&\leq  
c_2 \lambda^{t-t_0 } \|\psi (t_0 ) \|
+
c_2
\sum_{j=t_0 }^{t-1}  \lambda^{t-j-1 }
| w(j)|
+
\nonumber 
\\
&\qquad
c_2(\sqrt{\mu}+|r|)
\qquad t\geq t_0+n+1.
\end{flalign*}
Using the fact that $y(t)=\overline y(t)+r $ and 
since $\phi(t)$ is made up of $y(t)$ and $u(t)$ and their delayed
versions, we conclude that there exists $c_3$ so that 
\begin{flalign*}
 \|\phi(t) \|
&\leq  
c_3 \lambda^{t-t_0 } \|\psi (t_0 ) \|
+
c_3
\sum_{j=t_0 }^{t-1}  \lambda^{t-j-1 }
| w(j)|
+
\nonumber 
\\
&\qquad
c_3(\sqrt{\mu}+|r|),
\qquad t\geq t_0+2n+1.
\end{flalign*}
Observe that 
$\|\psi(t_0) \|\leq 2\|\boldsymbol\phi_0\| + (n+1)|r| $, 
so substitute this into the above 
bound as well; 
finally, since nothing unusual can happen in the first $2n$ steps, we
can extend the bound on $\phi(t) $ back to $t_0$ by adjusting the constant $c_3$
appropriately to 
yield the desired bound on $\phi(t)$ concluding the proof of Part (i)
of Theorem \ref{theorem_main}.

\underline{Proving Part ii) of Theorem \ref{theorem_main}.}
We now proceed to prove asymptotic tracking for the case when $w(\cdot) $ is constant, i.e. $\overline w(t)=0 $ for all $t$.
We would like first to show the asymptotic convergence of the prediction error $e(\cdot) $.
From \eqref{thmX_claim1}, we can establish the boundedness of $\psi(\cdot) $ in this case:
\begin{equation}
  \|\psi(t) \| \leq 
\gamma(\|\psi(t_0) \| + \sqrt{\mu}), \quad t\geq t_0.
\end{equation}
So from Proposition \ref{est_prop}(i) we can obtain 
\begin{equation}
  \sum_{j=t_0}^{t-1}
  \frac{e(j+1)^2}{\mu + \|\psi(j) \|^2}
  \leq  
  2\|\tilde\theta(t_0) \|^2 
  \leq 
  2 \overline{\mathbf s}^2, \quad t\geq t_0,
\end{equation}
so we have 
\begin{flalign}
  \sum_{j=t_0}^{t-1}
  e(j+1)^2
  &\leq  
    (2 \overline{\mathbf s}^2)
  (\mu + \sup_{j\geq t_0} \|\psi(j) \|^2)
  \nonumber 
  \\
  &\leq  
    (2 \overline{\mathbf s}^2)
  ((1+2\gamma^2)\mu +2 \gamma^2\|\psi(t_0) \|^2)
    , \;\; t\geq t_0.
    \nonumber
\end{flalign}
This shows that $e(\cdot) $ is square-summable; furthermore, from \eqref{pred_tilde1}
we know that $e(\cdot)$ is bounded because of the boundedness of both $\tilde\theta $ and $\psi $; hence $e(t)\rightarrow 0 $, as $t\rightarrow \infty $.

Now, from \eqref{est2} and the fact that projection cannot make the estimation worse, we can easily show that
\begin{equation}
\|\hat\theta(t+1) - \hat\theta(t)  \|^2  
\leq
\frac{\|\psi(t) \|^2}{(\mu +\|\psi(t) \|^2 )^2} e(t+1)^2
\leq 
\frac{\|\psi(t) \|^2}{\mu ^2} e(t+1)^2
\nonumber 
,
\end{equation}
so we can obtain 
\begin{flalign}
\sum_{j=t_0}^{t-1}\|\hat\theta(j+1) - \hat\theta(j)  \|^2  
&\leq 
\frac{\sup_{j\geq t_0}\|\psi(j) \|^2}{\mu ^2} 
\sum_{j=t_0}^{t-1}
e(j+1)^2.
\nonumber
\end{flalign}
From the boundedness of $\psi(\cdot) $ and the square-summability of 
$e(\cdot) $, we see from the above that $(\hat\theta(t+1) - \hat\theta(t) ) $ is 
square-summable as well. As $\hat\theta(\cdot) $ is bounded, we conclude also that $(\hat\theta(t+1) - \hat\theta(t)) \rightarrow 0 $ as $t\rightarrow\infty$.

Next, we have from \eqref{predict1},
$$e(t)=\overline y(t) - \psi(t-1)^\top \hat\theta(t-1),$$
so at every frozen $t> t_0$
it is easy to see that 
\[
E(z) = 
 \widehat {\overline {\mathbf A}}_{\hat{\theta}(t-1)}(z^{-1})
 \overline Y(z) -
 \widehat {\mathbf B}_{\hat{\theta}(t-1)}(z^{-1}) 
 \overline U(z);
\]
substituting $\overline U(z) $ with its value from the control law 
in \eqref{controlTF1}, 
\[
\overline U(z)=-\frac{\mathbf P_{\hat{\theta}(t-1)}(z^{-1})}{\mathbf L_{\hat{\theta}(t-1)}(z^{-1})} \overline Y(z),
\]
to conclude by utilizing \eqref{char1} that
\begin{flalign}
\mathbf L_{\hat{\theta}(t-1)}(z^{-1}) E(z)
&=
\bigl[\widehat {\overline {\mathbf A}}_{\hat{\theta}(t-1)}(z^{-1})
\mathbf L_{\hat{\theta}(t-1)}(z^{-1})
+
\nonumber 
\\
  &\qquad
\widehat {\mathbf B}_{\hat{\theta}(t-1)}(z^{-1})
\mathbf P_{\hat{\theta}(t-1)}(z^{-1})
\bigr] \overline Y(z)
\nonumber 
\\
\Rightarrow
\mathbf L_{\hat{\theta}(t-1)}(z^{-1}) E(z)
&=
\mathbf A^* ( z^{-1} ) \overline Y(z).
\label{y_track}
\end{flalign}
We can rewrite \eqref{y_track} as
\begin{equation}
e(t) + \sum_{j=1}^n  l_j(t-1) e(t-j) = \overline y(t) + \sum_{j=1}^{2n+1}  a^*_j \overline y(t-j);
\label{y_track2}
\end{equation}
let us write the LHS above as follows:
\begin{flalign}
 e(t) + \sum_{j=1}^n  &l_j(t-1) e(t-j)  
   = 
   e(t) + 
   \nonumber 
   \\ 
   & 
   \underbrace{\sum_{j=1}^n  l_j(t-1-j) e(t-j)}_{=:T_1(t)} + 
   \nonumber 
   \\ 
   & 
   \underbrace{\sum_{j=1}^n  [l_j(t-1) - l_j(t-1-j)] e(t-j)}_{=:T_2(t)}.
   \nonumber
   \end{flalign}
Since the coefficients $l_j(\cdot),\, j=1,2,\ldots n, $ are analytic functions of $\hat\theta(\cdot) \in\overline{\cal S} $, 
the coefficients $l_j(\cdot),\, j=1,2,\ldots n,$ are bounded,
and also
\begin{equation}
(l_j(t) - l_j(t-1)) \rightarrow 0    
\label{lj_conv}
\end{equation}
as $t\rightarrow\infty $
since $(\hat\theta(t+1) - \hat\theta(t)) \rightarrow 0 $
as $t \rightarrow \infty $.
Because $e(t)\rightarrow0 $ as $t\rightarrow\infty $ and because of the boundedness of $l_j(\cdot)$, we conclude that
for each $j=1,2,\ldots n $,
$$l_j(t-1-j) e(t-j) \rightarrow 0 $$ 
as $t\rightarrow\infty $, hence $T_1(t) \rightarrow 0 $ as $t\rightarrow\infty $.
Next, 
we have  
$l_j(t-1) - l_j(t-1-j) =  [l_j(t-1) - l_j(t-2)] + [l_j(t-2) - l_j(t-3)]
+ \cdots + [l_j(t-j) - l_j(t-1-j)] $, so
utilizing \eqref{lj_conv} 
we see that
$(l_j(t-1) - l_j(t-1-j)) \rightarrow 0 $ as $t\rightarrow\infty $ 
for every $j=1,2,\ldots n $;
with the boundedness of $e(\cdot) $, 
then we conclude that
for each $j=1,2,\ldots n $,
$$[l_j(t-1) - l_j(t-1-j)] e(t-j) \rightarrow 0 $$
as $t\rightarrow\infty $, hence $T_2(t) \rightarrow 0 $ 
as $t\rightarrow\infty $ as well. 
So it is clear that
\[
\left(e(t) + \sum_{j=1}^n l_j(t-1) e(t-j)  
\right)
\rightarrow 0
\]
as $t\rightarrow\infty $.
This means that for the RHS of \eqref{y_track2}, 
$$ 
\left(
\overline y(t) + \sum_{j=1}^{2n+1}  a^*_j \overline y(t-j)
\right)
\rightarrow 0
$$
as $t\rightarrow\infty $ as well.
Finally, since 
$ \mathbf A^* ( z^{-1} )$ has all its roots inside the unit disk and
its coefficients, $a^*_j,\, j=1,2,\dots,2n+1 $, are constant and bounded, then we conclude that the tracking error $\overline y(t) \rightarrow 0 $ as $t\rightarrow \infty $, as desired.
\end{proof}

\end{document}